\documentclass[10pt,letterpaper]{article} 

\usepackage{microtype}
\usepackage{amsmath,amsfonts,amsthm,amssymb,dsfont}
\usepackage{comment}
\usepackage{soul,color}
\usepackage{graphicx,float,wrapfig}
\usepackage{mathrsfs}
\usepackage[usenames,dvipsnames]{pstricks}
\usepackage{epsfig}
\usepackage{pst-grad} 
\usepackage{pst-plot} 
\usepackage[linesnumbered,boxed,ruled,vlined]{algorithm2e}
\usepackage{url}
\usepackage{fullpage}

\newtheorem{theo}{Theorem}[section]

\newtheorem{lem}[theo]{Lemma}

\newtheorem{prop}[theo]{Proposition}
\newtheorem{cor}[theo]{Corollary}
\newtheorem{defi}[theo]{Definition}
\newtheorem{rem}[theo]{Remark}

\newtheorem{problem}[theo]{Problem}
\newtheorem{question}[theo]{Question}

\newenvironment{proofof}[1]{\begin{proof}[Proof of #1]}{\end{proof}}

\newcommand{\distr}{\mathcal{D}}
\newcommand{\R}{\mathbb{R}}
\newcommand{\Ex}{\mathbb{E}}
\newcommand{\betw}{\ |\ }

\newcommand{\BPPPATH}{\mathsf{BPP_{path}}}

\newcommand{\SZK}{\textsf{SZK}}

\newcommand{\BPP}{\textsf{BPP}}

\newcommand{\SBP}{\textsf{SBP}}
\newcommand{\coSBP}{\textsf{coSBP}}

\newcommand{\SBQP}{\textsf{SBQP}}
\newcommand{\coSBQP}{\textsf{coSBQP}}

\newcommand{\AM}{\textsf{AM}}

\newcommand{\coAM}{\textsf{coAM}}
\newcommand{\PTIME}{\textsf{P}}
\newcommand{\NP}{\textsf{NP}}

\newcommand{\PH}{\textsf{PH}}
\newcommand{\PP}{\textsf{PP}}
\newcommand{\AzPP}{\textsf{A0PP}}
\newcommand{\coAzPP}{\textsf{coA0PP}}

\newcommand{\dt}{\textsf{dt}}

\newcommand{\PostBPP}{\textsf{PostBPP}}

\newcommand{\PostBQP}{\textsf{PostBQP}}

\newcommand{\Class}{\mathcal{C}}

\newcommand{\PSPACE}{\textsf{PSPACE}}
\newcommand{\PSPACEC}{\textsf{PSPACE-complete}}

\newcommand{\PostBPPdt}{\PostBPP^{\dt}}
\newcommand{\PostBQPdt}{\PostBQP^{\dt}}
\newcommand{\SBPdt}{\SBP^{\dt}}
\newcommand{\coSBPdt}{\coSBP^{\dt}}
\newcommand{\SBQPdt}{\SBQP^{\dt}}
\newcommand{\coSBQPdt}{\coSBQP^{\dt}}
\newcommand{\PPdt}{\PP^{\dt}}
\newcommand{\AzPPdt}{\AzPP^{\dt}}
\newcommand{\coAzPPdt}{\coAzPP^{\dt}}

\newcommand{\Ada}{\mathsf{Ada}}

\newcommand{\polylog}{\operatorname*{polylog}}

\newcommand{\Oracle}{\mathcal{O}}
\newcommand{\support}{\mathbf{support}}

\newcommand{\PTP}{\mathsf{PTP}}

\newcommand{\OR}{\mathsf{OR}}
\newcommand{\AND}{\mathsf{AND}}

\newcommand{\funcfami}{\mathbf}

\newcommand{\dtree}{\mathcal{T}}
\newcommand{\qalg}{\mathcal{Q}}

\newcommand{\adeg}{\widetilde{\mathrm{deg}}}
\newcommand{\dega}{\mathrm{deg}_{+}}
\newcommand{\mdega}{\widehat{\mathrm{deg}}_{+}}

\newcommand{\mdegs}{\widehat{\mathrm{deg}}_{-}}

\newcommand{\weight}{\mathsf{weight}}

\newcommand{\adaand}{\mathsf{AdaAND}}
\newcommand{\adaptp}{\mathsf{AdaPTP}}

\newcommand{\ddd}[3]{\distr_#1,\distr_{#2}^{d-1},\distr_{#3}^{d-1}}

\def\ShowAuthNotes{1}
\ifnum\ShowAuthNotes=1
\newcommand{\authnote}[2]{\ \\ \textcolor{red}{\parbox{0.9\linewidth}{[{\footnotesize {\bf #1:} { {#2}}}]}}\newline}
\else
\newcommand{\authnote}[2]{}
\fi

\title{Adaptivity vs Postselection}
\date{}
\author{
	Lijie Chen\thanks{This work was done when the author was visiting MIT.}\\
	\small Tsinghua University\\
	\small \texttt{wjmzbmr@gmail.com}
}

\begin{document}
	\maketitle
	
	\begin{abstract}
	We study the following problem: with the power of postselection (classically or quantumly), what is your ability to answer adaptive queries to certain languages? More specifically, for what kind of computational classes $\Class$, we have $\PTIME^{\Class}$ belongs to $\PostBPP$ or $\PostBQP$? While a complete answer to the above question seems impossible given the development of present computational complexity theory. We study the analogous question in {\em query complexity}, which sheds light on the limitation of {\em relativized} methods (the relativization barrier) to the above question.
	
	Informally, we show that, for a partial function $f$, if there is no efficient\footnote{In the world of query complexity, being efficient means using $O(\polylog(n))$ time.} {\em small bounded-error} algorithm for $f$ classically or quantumly, then there is no efficient {\em postselection bounded-error} algorithm to answer adaptive queries to $f$ classically or quantumly. Our results imply a new proof for the classical oracle separation $\PTIME^{\NP^{\Oracle}} \not\subset \PP^{\Oracle}$, which is arguably more elegant. They also lead to a new oracle separation $\PTIME^{\SZK^{\Oracle}} \not\subset \PP^{\Oracle}$, which is close to an oracle separation between $\SZK$ and $\PP$\textemdash an open problem in the field of oracle separations.
	
	Our result also implies a hardness amplification construction for polynomial approximation: given a function $f$ on $n$ bits, we construct an adaptive-version of $f$, denoted by $F$, on $O(m \cdot n)$ bits, such that if $f$ requires large degree to approximate to error $2/3$ in a certain one-sided sense, then $F$ requires large degree to approximate even to error $1/2 - 2^{-m}$. Our construction achieves the same amplification in the work of Thaler (ICALP, 2016), by composing a function with $O(\log n)$ {\em deterministic query complexity}, which is in sharp contrast to all the previous results where the composing {\em amplifiers} are all hard functions in a certain sense.  
	\end{abstract}
	\section{Introduction}

\subsection{Background}

The idea of postselection has been surprisingly fruitful in theoretical computer science and quantum computing~\cite{aaronson2011computational,drucker2009quantum,bremner2011classical}. Philosophically, it addresses the following question: if you believe in the Many-worlds interpretation\footnote{\url{https://en.wikipedia.org/wiki/Many-worlds_interpretation}} and can condition on a rare event (implemented by killing yourself after observing the undesired outcomes), then what would you be able to compute in a reasonable amount of time? The complexity classes $\PostBPP$~\cite{han1997threshold} and $\PostBQP$~\cite{aaronson2005quantum} are defined to represent the computational problems you can solve with the ability of postselection in a classical world or a quantum world. 

However, even with that seemingly omnipotent power of postselection, your computational power is still bounded. It is known that $\PostBPP \subseteq \PH$~\cite{han1997threshold}, and (surprisingly) $\PostBQP = \PP$~\cite{aaronson2005quantum}. Hence, it seems quite plausible that even with the postselection power, you are still not able to solve a $\PSPACEC$ problem, as it is widely believed that $\PH$ and $\PP$ are strictly contained in $\PSPACE$.

Another more non-trivial (and perhaps unexpected) weakness of those postselection computation classes, is their inability to simulate {\em adaptive queries} to certain languages. For example, it is known that $\PTIME^{\NP[O(\log n)]}$\footnote{$O(\log n)$ stands for the  $\PTIME$ algorithm can only make $O(\log n)$ queries to the oracle.} is contained in $\PostBPP$~\cite{han1997threshold}, and this result relativizes. But there is an oracle separation between $\PTIME^{\NP[\omega(\log n)]}$ and $\PostBQP$~\cite{beigel1994perceptrons}. In other words, there is no relativized $\PostBQP$ algorithm that can simulate $\omega(\log n)$ adaptive queries to a certain language in $\NP$. In contrast, we know that $\PTIME^{\parallel \NP} \subseteq \PostBPP \subseteq \PP$ \cite{han1997threshold}, hence they are capable of simulating {\em non-adaptive} queries to $\NP$.

Then a natural question follows:

\begin{question}\label{ques:comp-ques}
What is the limit of the abilities of these postselection classes on simulating adaptive queries to certain languages? More specifically, is there any characterization of the complexity class $\Class$ such that $\PTIME^{\Class}$ is contained in $\PostBPP$ or $\PostBQP$?
\end{question}

Arguably, a complete answer to this problem seems not possible at the present time: even determining whether $\PTIME^{\NP} \subseteq \PP$ is already extremely hard, as showing $\PTIME^{\NP} \subseteq \PP$ probably requires some new non-relativized techniques, and proving $\PTIME^{\NP} \not\subset \PP$ implies $\PH \not\subset \PP$, which is a long-standing open problem.

\subsection{Relativization and the analogous question in query complexity}

So in this paper, inspired by the oracle separation in~\cite{beigel1994perceptrons}, we study this problem from a relativization point of view. {\em Relativization}, or {\em oracle separations} are ultimately about the {\em query complexity}. Given a complexity class $\Class$, there is a canonical way to define its analogue in query complexity: partial functions which are computable by a {\em non-uniform} $\Class$ machine with $\polylog(n)$ queries to the input. For convenience, we will use $\Class^{\dt}$ to denote the query complexity version of $\Class$. We adopt the convention that  $\Class^{\dt}$ denotes the query analogue of $\Class$, while $\Class^{\dt}(f)$ denotes the $\Class^{\dt}$ complexity of the partial function $f$.

\newcommand{\inputlen}{\mathsf{len}}

For a partial function $f$, we use $\inputlen(f)$ to denote its input length. We say a family of partial functions $\funcfami{f} \in \Class^{\dt}$, if $\Class^{\dt}(f) = O(\polylog(\inputlen(f)))$ for all $f \in \funcfami{f}$.

In order to study this question in the query complexity setting, given a partial function $f$, we need to define its adaptive version.

\begin{defi}[Adaptive Construction]
	Given a function $f :D \to \{0,1\}$ with $D \subseteq \{0,1\}^M$ and an integer $d$, we define $\Ada_{f,d}$, its depth $d$ adaptive version, as follows:
	
	$$
	\begin{array}{ c c }
	\Ada_{f,0} := f \quad \text{and} & 
	\begin{array}{ c }
	\Ada_{f,d}: D \times D_{d-1} \times D_{d-1} \to \{0,1\}\\
	\Ada_{f,d}(w,x,y) := \begin{cases}
	\Ada_{f,d-1}(x) & \quad f(w) = 0\\
	\Ada_{f,d-1}(y) & \quad f(w) = 1\\
	\end{cases}
	\end{array}
	\end{array}
	$$
	where $D_{d-1}$ denotes the domain of $\Ada_{f,d-1}$. 
	
	The input to $\Ada_{f,d}$ can be encoded as a string of length $(2^{d+1}-1)\cdot M$. Thus, $\Ada_{f,d}$ is a partial function from $D^{(2^{d+1}-1)} \to \{0,1\}$.
\end{defi}

Then, given a family of partial function $\funcfami{f}$, we define $\Ada_{\funcfami{f}} := \{ \Ada_{f,d} \betw f \in \funcfami{f}, d \in \mathbb{N} \}$.

Notice that when you have the ability to adaptively solve $d+1$ queries to $f$ (or with high probability), then it is easy to solve $\Ada_{f,d}$.
Conversely, in order to solve $\Ada_{f,d}$, you need to be able to adaptively answer $d+1$ questions to $f$, as even {\em knowing what is the right $i^{th}$ question to answer} requires you to correctly answer all the previous $i-1$ questions.

Now, everything is ready for us to state the analogous question in query complexity.

\begin{question}\label{ques:query-ques}
	What is the characterization of the partial functions family $\funcfami{f}$ such that $\Ada_{\funcfami{f}} \in \PostBPPdt$ ($\PostBQPdt$)?
\end{question}

There are at least two reasons to study Question~\ref{ques:query-ques}. First, it is an interesting question itself in query complexity. Second, an answer to Question~\ref{ques:query-ques} also completely characterizes the limitation on the {\em relativized techniques} for answering Question~\ref{ques:comp-ques}, i.e., the limitation of relativized methods for simulating {\em adaptive queries} to certain complexity classes with the power of postselection.

This paper provides some interesting results toward resolving Question~\ref{ques:query-ques}.

\subsection{Our results}

Despite that we are not able to give a complete answer to Question~\ref{ques:query-ques}. We provide some interesting lower bounds showing that certain functions' adaptive versions are hard for these postselection classes.

Formally, we prove the following two theorems.

\begin{theo}[Quantum Case]\label{theo:ada-postbqp}
	For a family of partial function $\funcfami{f}$, $\Ada_{\funcfami{f}} \not\in \PostBQPdt(\PPdt)$ if $\funcfami{f} \not\in \SBQPdt \cap \coSBQPdt$.
\end{theo}

\begin{theo}[Classical Case]\label{theo:ada-postbpp}
	For a family of partial function $\funcfami{f}$, $\Ada_{\funcfami{f}} \not\in \PostBPPdt$ if $\funcfami{f} \not\in \SBPdt \cap \coSBPdt$.
\end{theo}

Roughly speaking, $\SBP$ is a relaxation of $\BPP$, it is the set of languages $L$ such that there exists a $\BPP$ machine $M$, which accepts $x$ with probability $\ge 2\alpha$ if $x \in L$; and with probability $\le \alpha$ if $x \not\in L$ for a positive real number $\alpha$. And $\SBQP$ is the quantum analogue of $\SBP$, where you are allowed to use a polynomial time quantum algorithm instead.\footnote{For the formal definitions of $\SBP$, $\PostBPP$, $\PostBQP$, $\SBQP$ and their equivalents in query complexity, see the preliminaries.}

Our theorems show that, for a partial function $f$, if there is no efficient classical (quantum) algorithm which accepts all the $1$-inputs with a slightly better chance than all the $0$-inputs, then there is no efficient $\PostBPP$ ($\PostBQP$) algorithm that can answer adaptive queries to $f$.

In fact, we prove the following two quantitatively tighter theorems, from which Theorem~\ref{theo:ada-postbqp} and Theorem~\ref{theo:ada-postbpp} follows easily.

\begin{theo} \label{theo:hard-adapt-postbqp}
	Let $f$ be a partial function and $T$ be a non-negative integer. Suppose $\mdega(f) > T$ or $\mdegs(f) > T$, then we have 
	$$
	\PPdt(\Ada_{f,d}) > \min(T/4,2^{d-1}).\footnote{ $\mdega$ ({\em one-sided low-weight approximate degree}, cf. Definition~\ref{defi:low-weight}) is equivalent to $\SBQPdt$, and $\PPdt$ is equivalent to $\PostBQPdt$, see Corollary~\ref{cor:eq-SBQP-AZPP}, Corollary~\ref{cor:eq-PostBQP-PP} and Theorem~\ref{theo:eq-AZPP-mdega}. We state it in this form for simplicity.}
	$$
\end{theo}

\begin{theo} \label{theo:hard-adapt-postbpp}
	Let $f:D \to \{0,1\}$ with $D \subseteq \{0,1\}^M$ be a partial function and $d$ be a non-negative integer. Suppose $\SBPdt(f) > T$ or $\coSBPdt(f) > T$, then we have 
	$$
	\PostBPPdt(\Ada_{f,d}) > \min(T/5,(2^{d}-1)/5).
	$$
\end{theo}

\subsection{Applications in oracle separations}

Our results have several applications in oracle separations.
\begin{itemize}
	\item A new proof for $\PTIME^{\NP^{\Oracle}} \not\subset \PP^{\Oracle}$:
	
	We prove that $\SBQPdt(f)$ is indeed equivalent to {\em one-sided low-weight approximate degree}, denoted by $\mdega(f)$ (cf. Definition~\ref{defi:low-weight}), which is lower bounded by one-sided approximate degree $\dega(f)$ (cf. Definition~\ref{defi:deg-approx}). 
	
	Using the fact that $\dega(\AND_n) \ge \Omega(\sqrt{n})$, Theorem~\ref{theo:ada-postbqp} implies that $\Ada_{\funcfami{\AND}} \not\subset \PPdt$, yielding a simpler proof for the classical oracle separation between $\PTIME^{\NP}$ and $\PP$ in \cite{beigel1994perceptrons}. 
	
	Our proof is arguably simpler and more elegant. Also, unlike the seemingly artificial problem $\textsf{ODD-MAX-BIT}$\footnote{Given a binary input $x$, it asks whether the rightest $1$ in $x$ is in an odd position.} in \cite{beigel1994perceptrons}, $\Ada_{\funcfami{\AND}}$ looks like a more natural hard problem in $\PTIME^{\NP}$.
	\item The new oracle separation $\PTIME^{\SZK^{\Oracle}} \not\subset \PP^{\Oracle}:$
	
    Since the {\em Permutation Testing Problem}, denoted by $\PTP_n$ (see Problem~\ref{prob:PTP} for a formal definition), satisfies $\dega(\PTP_n) \ge \Omega(n^{1/3})$ and has a $\log(n)$-time $\SZK$ protocol. Theorem~\ref{theo:ada-postbqp} implies that $\Ada_{\PTP} \not\subset \PPdt$, which in turn shows an oracle separation between $\PTIME^{\SZK}$ and $\PP$.
    
    It has been an open problem~\cite{aaronson2012impossibility} that whether there exists an oracle separation between $\SZK$ and $\PP$, our result is pretty close to an affirmative answer to that.\footnote{Partially inspired by this work, an oracle separation between $\SZK$ and $\PP$ (in fact, $\mathsf{UPP}$) has been constructed in a very recent work of Bouland, Chen, Holden, Thaler and Vasudevan~\cite{bouland2016szk}, thus resolved this open problem.} 
    
    Also, note that $\PTIME^{\SZK} \subseteq \PTIME^{\AM \cap \coAM} = \AM \cap \coAM$,  so our result improves on the oracle separation between $\AM \cap \coAM$ and $\PP$ by Vereschchagin~\cite{vereschchagin1992power}.
\end{itemize}

\subsection{Applications in hardness amplification for polynomial approximation}

Our construction also leads to a hardness amplification theorem for polynomial approximation. In order to state our result, we need to introduce the definition of two approximate degrees first.

\begin{defi}\label{defi:deg-approx}
	The $\epsilon$-approximate degree of a partial function of $f: D \to \{0,1\}$, denoted as $\adeg_\epsilon(f)$, is the least degree of a real polynomial $p$ such that $|p(x)-f(x)| \le \epsilon$ when $x \in D$, and $|p(x)| \le 1+\epsilon$ when $x\not\in D$. 
	
	We say a polynomial $p$	one-sided $\epsilon$-approximates a partial Boolean function $f$, if $p(x) \in [0,\epsilon]$ when $f(x) = 0$, and $p(x) \ge 1$ when $f(x) = 1$.\footnote{Our definition of one-sided approximation is slightly different from the standard one~\cite{sherstov2014breaking,bun2015dual,sherstov2015power}, but it greatly simplifies several discussions in our paper, and they are clearly equivalent up to a linear transformation in $\epsilon$.} Then the one-sided $\epsilon$-approximate degree of a partial function $f$, denoted by $\dega^{\epsilon}(f)$, is the minimum degree of a polynomial one-sided $\epsilon$-approximating $f$.
\end{defi}

Now we are in a position to state our amplification theorem.

\begin{theo}\label{theo:ampl}
	Let $f$ be a partial function such that $\dega^{2/3}(f) > T$ and $d$ be a positive integer, we have $\adeg_{\epsilon}(\Ada_{f,d}) > T$ for $\epsilon = 0.5-2^{-2^d+1}$. 
\end{theo}

\newcommand{\AdaQ}{\mathsf{AdaQ}}

That is, given a function with high one-sided approximate degree for an error constant bounded away from $1$, it can be transformed to a function with high approximate degree even for $\epsilon$ {\em doubly exponentially} close to $1/2$ in $d$.\footnote{Which is {\em single exponential} in the input length of the {\em amplifier} $\AdaQ$, see the discussion below.}

\paragraph*{Comparison with previous amplification results}

There have been a lot of research interest in hardness amplification for polynomial approximation, many amplification results are achieved through {\em function composition}~\cite{bun2015hardness,sherstov2014breaking,thaler2014lower}. We use $f \circ g$ to denote the block composition of $f$ and $g$, i.e. $f(g,g,\dotsc,g)$.

Our result can also be viewed as one of them. Let $\AdaQ_{d} := \Ada_{\mathsf{id},d}$, where $\mathsf{id}$ is just the identity function from $\{0,1\}$ to $\{0,1\}$. Then we can see that in fact $\Ada_{f,d}$ is equivalent to $\AdaQ_{d} \circ f$. Let $n = 2^{d+1} - 1$, which is the input length of $\AdaQ_{d}$.

However, all the previous amplification results are achieved by letting the {\em amplifier} $f$ to be a {\em hard} function. We list all these results for an easy comparison.

\begin{itemize}
	\item In the work of Bun and Tahler~\cite{bun2015hardness}, they showed that for a function $g$ such that $\dega(g) > T$, $\adeg_\epsilon(\OR_n \circ g) > T$ for $\epsilon = 1/2 - 2^{-\Omega(n)}$. This is further improved by Sherstov~\cite{sherstov2014breaking} to that $\mathrm{deg}_{\pm}(\OR_n \circ g) = \Omega(\min(n,T))$. Here, the amplifier $\OR_n$ is a hard function in the sense that $\dega(\OR_n) \ge \Omega(\sqrt{n})$~\cite{nisan1994degree}.
	\item In~\cite{thaler2014lower}, Thaler showed that for a function $g$ such that $\dega(g) > T$, $\adeg_\epsilon(\textsf{ODD-MAX-BIT}_n \newline\circ g) > T$ for $\epsilon = 1/2 - 2^{-\Omega(n)}$.\footnote{This construction is further improved in a very recent work~\cite{bun2016approximate} by Bun and Thaler, with a more sophisticated construction which does not follow the composition paradigm.} In this case, the amplifier $\textsf{ODD-MAX-BIT}_n$ is even harder in the sense that it has a $\PPdt$ query complexity of $\Omega(\sqrt[3]{n})$~\cite{beigel1994perceptrons}.
	\item Moreover, it is easy to see that the {\em randomized query complexity} of both $\OR_n$ and $\textsf{ODD-MAX-BIT}_n$ is the maximum possible $\Omega(n)$.
\end{itemize}

In contrast, our amplifier $\AdaQ$, is {\em extremely} simple\textemdash it has a {\em deterministic query complexity} of $O(\log n)$!\footnote{A simple $O(\log n)$-query algorithm just follows from the definition.}

This is a rather surprising feature of our result. That means $\AdaQ$ also has an {\em exact degree} of $O(\log n)$.
Intuitively, composing with such a simple and innocent function seems would not affect the hardness of the resulting function. Our result severely contradicts this intuition. But from the view point of Theorem~\ref{theo:ada-postbqp}, composing with $\AdaQ$ indeed ``adaptivize'' the function, makes it hard for \PostBQP\ algorithms, which is in turn closely connected to $\PP$ algorithms and therefore polynomial approximate degree. So this result is arguably natural under that perspective, which illustrates a recurring theme in TCS: a new perspective can lead to some unexpected results.

\subsection{Paper organization}

In Section~\ref{sec:pre} we introduce some preliminaries, due to the space constraints, some of the formal definitions of those partial function classes in query complexity can be found in the appendix.
We prove Theorem~\ref{theo:ada-postbqp} and Theorem~\ref{theo:hard-adapt-postbqp} in Section~\ref{sec:quantum}, and defer the proof for Theorem~\ref{theo:ada-postbpp} and Theorem~\ref{theo:hard-adapt-postbpp} to the appendix. Theorem~\ref{theo:ampl} is proved in Section~\ref{sec:amply}. And we provide formal proofs for the two oracle separation results in the appendix.
	\section{Preliminaries}\label{sec:pre}

\subsection{Decision trees and quantum query algorithms}

A (randomized) decision tree is the analogue of a deterministic (randomized) algorithm in the query complexity world, and a quantum query algorithm is the analogue of a quantum algorithm. See \cite{buhrman2002complexity} for a nice survey on query complexity.

Let $\dtree$ be a randomized decision tree, we use $\mathcal{C}(\dtree)$ to denote the maximum number of queries incurred by $\mathcal{T}$ in the worst case\footnote{i.e. the maximum height of a decision tree in the support of $\dtree$}. Let $\qalg$ be a quantum query algorithm, we use $\mathcal{C}(\qalg)$ to denote the number of queries taken by $\qalg$.

We assume a randomized decision tree $\dtree$ (or a quantum query algorithm $\qalg$) outputs a result in $\{0,1\}$, and we use $\dtree(x)$ ($\qalg(x)$) to denote the (random) output of $\dtree$ ($\qalg$) given an input $x$.

\subsection{Complexity classes and their query complexity analogues}

We assume familiarity with some standard complexity classes like \PP. Due to space constraint, we only introduce the most relevant classes $\AzPPdt$ and $\PPdt$ here, and defer the formal definitions of the partial function complexity classes $\SBPdt$, $\SBQPdt$, $\PostBPPdt$ and $\PostBQPdt$ to the appendix.

Recall that $\Class^{\dt}$ is the set of the partial function family $\funcfami{f}$ with $\Class^{\dt}(f) = O(\polylog(\inputlen(f)))$ for all $f \in \funcfami{f}$, hence we only need to define $\Class^{\dt}(f)$ for a partial function $f$.

\paragraph*{$\PPdt$} 

We first define $\PPdt(f)$.

\begin{defi}\label{defi:PPdt}
	Let $f:D \to \{0,1\}$ with $D \subseteq \{0,1\}^M$ be a partial function. Let $\dtree$ be a randomized decision tree which computes $f$ with a probability better than $1/2$. Let $\alpha$ be the maximum real number such that
	$$
	\Pr[\dtree(x) = f(x)] \ge \frac{1}{2} + \alpha
	$$
	for all $x \in D$.
	
	Then we define $\PPdt(\dtree;f) := C(\dtree) + \log_2(1/\alpha)$, and $\PPdt(f)$ as the minimum of $\PPdt(\dtree;f)$ over all $\dtree$ computing $f$ with a probability better than $1/2$.
\end{defi}

\paragraph*{$\AzPP$ and $\AzPPdt$}
In this subsection we review the definition of $\AzPP$, and define its analogue in query complexity. There are several equivalent definitions for $\AzPP$, we choose the most convenient one here.

\begin{defi}
	$\AzPP$ (defined by Vyalyi~\cite{vyalyi2003qma}) is the class of languages $L\subseteq\left\{
	0,1\right\}  ^{\ast}$\ for which there exists a $\mathsf{BPP}$\ machine $M$ and a polynomial $p$, such that for all inputs $x$:
	\begin{enumerate}
		\item[(i)] $x\in L\Longrightarrow\Pr\left[  M\left(  x\right)  \text{ accepts}\right]\geq \frac{1}{2} + 2^{-p(|x|)}$.
		
		\item[(ii)] $x\notin L\Longrightarrow\Pr\left[  M\left(  x\right)  \text{
			accepts}\right]  \in \left[\frac{1}{2}, \frac{1}{2} + 2^{-p(|x|)-1} \right]$.
	\end{enumerate}
\end{defi}

\begin{defi}\label{defi:AzPPdt}
	
	Let $f:D \to \{0,1\}$ with $D \subseteq \{0,1\}^M$ be a partial function. We say a randomized decision tree $\dtree$ $\AzPP$-computes $f$ if there is a real number $\alpha > 0$ such that
	\begin{itemize}
		\item $\Pr[\dtree(x)=1] \ge 1/2 + 2\alpha$ when $f(x) = 1$.
		\item $\Pr[\dtree(x)=1] \in [1/2, 1/2 + \alpha]$ when $f(x) = 0$.
	\end{itemize}
	Fix a $\dtree$ $\AzPP$-computing $f$, let $\alpha$ be the maximum real number satisfying above conditions. Then we define $\AzPPdt(\dtree;f) = C(\dtree) + \log_2(1/\alpha)$ for $\dtree$ $\AzPP$-computing $f$ and $\AzPPdt(f)$ as the minimum of $\AzPPdt(\dtree;f)$ over all $\dtree$ $\AzPPdt$-computing $f$. And we simply let $\coAzPPdt(f) := \AzPPdt(\neg f)$.
\end{defi}

\paragraph*{Two relativized facts}

We also introduce two important relativized results here.
In \cite{aaronson2005quantum}, Aaronson showed that $\PostBQP$ is indeed $\PP$ in disguise.

\begin{theo}[\cite{aaronson2005quantum}]
	$\PostBQP = \PP$.
\end{theo}

And in~\cite{kuperberg2009hard}, Kuperberg showed that $\SBQP$ is in fact equal to $\AzPP$.

\begin{theo}[\cite{kuperberg2009hard}]
	$\SBQP = \AzPP$.
\end{theo}

These two theorems relativize, hence we have the following corollaries.

\begin{cor}\label{cor:eq-SBQP-AZPP}
	$\SBQPdt = \AzPPdt$.
\end{cor}

\begin{cor}\label{cor:eq-PostBQP-PP}
	$\PostBQPdt = \PPdt.$
\end{cor}
\subsection{Low-weighted one-sided approximate degree}

In this subsection, we introduce a new notion of one-sided approximate degree, which is closely connected to $\AzPPdt(f)$.

\begin{defi} \label{defi:low-weight}
Write a polynomial $p(x) := \sum_{i=1}^{m} a_i \cdot M_i(x)$ as a sum of monomials, we define $\weight(p) := \sum_{i=1}^{m} |a_i|$. The one-sided low-weight $\epsilon$-approximate degree of a partial function $f$ denoted by $\mdega^{\epsilon}(f)$, is defined by
$$
\mdega^{\epsilon}(f) := \min_{p} \max \{\mathrm{deg}(p),\log_2(\weight(p)) \},
$$
where $p$ goes over all polynomials which one-sided $\epsilon$-approximates $f$.\footnote{Recall that a polynomial $p$ one-sided $\epsilon$-approximates a partial Boolean function $f$, if $p(x) \in [0,\epsilon]$ when $f(x) = 0$, and $p(x) \ge 1$ when $f(x) = 1$ as in Definition~\ref{defi:deg-approx}.} 

We simply let $\mdegs^{\epsilon}(f) := \mdega^{\epsilon}(\neg f)$. We also define $\mdega(f)$ as $\mdega^{1/2}(f)$. $\mdegs$ is defined similarly.
\end{defi}

Clearly $\mdega^{\epsilon}(f) \ge \dega^{\epsilon}(f)$. And the choice of constant $1/2$ is arbitrary, as we can reduce the approximation error by the following lemma.

\begin{lem}\label{lm:amply}
	For any $0 < \epsilon_1 < \epsilon_2 < 1$, $\mdega^{\epsilon_1}(f) \le \left\lceil \frac{ \ln \epsilon_1^{-1}}{ \ln \epsilon_2^{-1} } \right\rceil \cdot \mdega^{\epsilon_2}(f)$.
\end{lem}
\begin{proof}
	We can just take the $\left\lceil \frac{ \ln \epsilon_1^{-1}}{ \ln \epsilon_2^{-1} } \right\rceil^{th}$ power of the polynomial corresponding to $\mdega^{\epsilon_2}(f)$.
\end{proof}

We show that $\mdega(f)$ is in fact equivalent to $\AzPPdt(f)$ up to a constant factor.

\begin{theo}\label{theo:eq-AZPP-mdega}
Let $f$ be a partial function, then
$$
\mdega(f) \le 2 \cdot \AzPPdt(f) \text{ and } \AzPPdt(f) \le 2 \cdot \mdega(f) + 2.
$$ 
\end{theo}

The proof is based on a simple transformation between a decision tree and the polynomial representing it, we defer the details to the appendix. 

And the following corollary follows from the definitions.

\begin{cor}\label{cor:mdegs-bound}
	Let $f$ be a partial function, then
	$$
	\mdegs(f) \le 2 \cdot \coAzPPdt(f)  \text{ and } \coAzPPdt(f) \le 2 \cdot \mdegs(f) + 2.
	$$ 
\end{cor}

\subsection{The permutation testing problem}

Finally, we introduce the permutation testing problem.

\begin{problem}[Permutation Testing Problem or PTP]\label{prob:PTP} Given black-box access to a function
		$f:\left[  n\right]  \rightarrow\left[  n\right]  $, and promised that either
		
		\begin{enumerate}
			\item[(i)] $f$ is a permutation (i.e., is one-to-one), or
			
			\item[(ii)] $f$ differs from every permutation on at least $n/8$\ coordinates.
		\end{enumerate}
		
		The problem is to accept if (i) holds and reject if (ii) holds.
		
		Assume $n$ is a power of $2$, we use $\PTP_n$ to denote the Permutation Testing Problem on functions from $[n] \to [n]$. $\PTP_n$ can be viewed as a partial function $D \to \{0,1\}$ with $D \subseteq \{0,1\}^{n \cdot \log_2 n}$.
\end{problem} 
	\section{Proof for the quantum case}\label{sec:quantum}

In this section we prove Theorem~\ref{theo:ada-postbqp}.

Let $f:D \to \{0,1\}$ with $D \subseteq \{0,1\}^M$ be a partial function, we say a polynomial $p$ on $M$ variables {\em computes} $f$, if $p(x) \ge 1$ whenever $f(x)=1$, and $p(x) \le -1$ whenever $f(x)=0$.

\subsection{Existence of the hard distributions}
In this subsection we show that if $\mdega(f)$ is large, there must exist some input distributions witness this fact in a certain sense.

\begin{lem}\label{lm:single-dist-exist-mdega}
Let $f$ be a partial function and $T$ be a non-negative integer. For convenience, we say a polynomial $p$ is {\em valid}, if it is of degree at most $T$, and satisfies $\weight(p) \le 2^{T}$. 

If $\mdega^{2/3}(f) > T$, there exist two distributions $\distr_0$ and $\distr_1$ supported on $f^{-1}(0)$ and $f^{-1}(1)$ respectively, such that
$$
- p(\distr_0) > 2 \cdot p(\distr_1),
$$
where $p(\distr) = \Ex_{x \sim \distr}[p(x)]$, for all valid polynomial $p$ computing $f$.
\end{lem}

In order to establish the above lemma, we need the following simple lemma. 

\begin{lem}\label{lm:help-1}
For any valid polynomial $p$ computing $f$, if $\mdega^{2/3}(f) > T$, then there exist $x \in f^{-1}(0)$ and $y \in f^{-1}(1)$ such that $-p(x) > 2 \cdot p(y)$.
\end{lem}

The proof is based on a simple calculation, the details can be found in the appendix.

Then we prove Lemma~\ref{lm:single-dist-exist-mdega}.

\begin{proofof}{Lemma~\ref{lm:single-dist-exist-mdega}}
By Lemma~\ref{lm:help-1}, we have
$$
\min_{\text{$p$}} \max_{(x,y) \in f^0 \times f^1} -p(x) - 2 \cdot p(y) > 0,
$$
where $p$ is a valid polynomial which computes $f$, $f^0 := f^{-1}(0)$ and $f^1 := f^{-1}(1)$. By the minimax theorem, and note that all the valid polynomials form a {\em compact convex set}, there exists a distribution $\distr_{xy}$ on $f^0 \times f^1$ such that for any valid polynomial $p$ computing $f$, we have
$$
\Ex_{(x,y) \sim \distr_{xy}}[-p(x)-2 \cdot p(y)] > 0.
$$
Then we simply let $\distr_0$ ($\distr_1$) be the marginal distribution of $\distr_{xy}$ on $f^0$ ($f^1$), which completes the proof.
\end{proofof}

And the following corollary follows by the definition of $\mdegs$.

\begin{cor}\label{cor:single-dist-exist-mdegs}
	Let $f$ be a partial function and $T$ be a non-negative integer, if $\mdegs^{2/3}(f) > T$, then there exist two distributions $\distr_0$ and $\distr_1$ supported on $f^{-1}(0)$ and $f^{-1}(1)$ respectively, such that for all valid polynomial $p$ computing $f$,
	$$
	p(\distr_1) > -2 \cdot p(\distr_0).
	$$
\end{cor}

\subsection{Proof for Theorem~\ref{theo:ada-postbqp} and Theorem~\ref{theo:hard-adapt-postbqp}}

We first show Theorem~\ref{theo:hard-adapt-postbqp} implies Theorem~\ref{theo:ada-postbqp}.

\begin{proofof}{Theorem~\ref{theo:ada-postbqp}}
	Suppose $\funcfami{f} \not\in \SBQPdt$, the case that $\funcfami{f} \not\in \coSBQPdt$ is similar.
	
	By Corollary~\ref{cor:eq-SBQP-AZPP} and Theorem~\ref{theo:eq-AZPP-mdega}, there exists a sequence of function $\{f_i\}_{i=1}^{\infty} \subseteq \funcfami{f}$ such that $\mdega(f_i) > \log(\inputlen(f_i))^{i}$. Then we consider the partial function sequence $\{ \Ada_{f_i,\lceil\log(\inputlen(f_i)) \rceil} \}_{i=1}^{\infty} \subseteq \Ada_{\funcfami{f}}$.
	
	By Theorem~\ref{theo:hard-adapt-postbqp}, we have 
	$$
	\PPdt(\Ada_{f_i,\lceil\log(\inputlen(f_i)) \rceil}) > \min(\log(\inputlen(f_i))^{i}/4,\inputlen(f_i)/2).
	$$
	Note that $\inputlen(\Ada_{f_i,\lceil\log(\inputlen(f_i)) \rceil}) \le 2 \cdot \inputlen(f_i)^2$, we can see $\Ada_{\funcfami{f}} \notin \PPdt$ due to the above partial function sequence.
\end{proofof}

Now, we are going to prove Theorem~\ref{theo:hard-adapt-postbqp}. We begin by introducing some consequences of a function having low $\PPdt$ complexity.

\begin{lem}\label{lm:pps-cons}
	Let $f$ be a partial function, $T$ be a positive integer. Suppose $\PPdt(f) \le T$, then there exists a degree $T$-polynomial $p$ computing $f$ and satisfying $\weight(p) \le 2^{2T}$.
\end{lem}

The proof is based on a direct analysis of the polynomial representing the decision tree for $\PPdt(f)$, we defer the details to the appendix.

Our proof relies on the following two key lemmas.

\begin{lem}\label{lm:dist-exist-mdega}
	Let $f$ be a partial function with $\mdega^{2/3}(f) > T$. Then for each integer $d$, there exist two distributions $\distr_1^d$ and $\distr_0^d$ supported on $\Ada_{f,d}^{-1}(1)$ and $\Ada_{f,d}^{-1}(0)$ respectively, such that $- p(\distr_0) > 2^{2^d} \cdot p(\distr_1)$ for any degree-$T$ polynomial $p$ computing $\Ada_{f,d}$ and satisfying $\weight(p) \le 2^{T}$.
\end{lem}

\begin{lem}\label{lm:dist-exist-mdegs}
	Let $f$ be a partial function with $\mdegs^{2/3}(f) > T$. Then for each integer $d$, there exist two distributions $\distr_1^d$ and $\distr_0^d$ supported on $\Ada_{f,d}^{-1}(1)$ and $\Ada_{f,d}^{-1}(0)$ respectively, such that $p(\distr_1) > -2^{2^d} \cdot p(\distr_0)$ for any degree-$T$ polynomial $p$ computing $\Ada_{f,d}$ and satisfying $\weight(p) \le 2^{T}$.
\end{lem}

We first show these two lemmas imply Theorem~\ref{theo:hard-adapt-postbqp} in a straightforward way.

\begin{proofof}{Theorem~\ref{theo:hard-adapt-postbqp}}
We prove the case when $\mdega(f) > T$ first.

Otherwise, suppose $\PPdt(\Ada_{f,d}) \le \min(T/4,2^{d-1})$. By Lemma~\ref{lm:pps-cons}, we have a degree-$T/4$ polynomial $p$ computing $\Ada_{f,d}$ with $\weight(p) \le \min(2^{T/2},2^{2^d})$. From Lemma~\ref{lm:amply}, $\mdega(f)=\mdega^{1/2}(f) \le 2 \cdot \mdega^{2/3}(f)$, hence $\mdega^{2/3}(f) > T/2$. Then by Lemma~\ref{lm:dist-exist-mdega}, there exist two distributions $\distr_1^d$ and $\distr_0^d$ supported on $\Ada_{f,d}^{-1}(1)$ and $\Ada_{f,d}^{-1}(0)$ respectively, such that $- p(\distr_0) > 2^{2^d} \cdot p(\distr_1)$ as $p$ is of degree at most $T/4$ and satisfies $\weight(p) \le 2^{T/2}$.

But this means that $-p(\distr_0) > 2^{2^d}$, which implies there exists an $x$ such that $p(x) < -2^{2^d}$, therefore $\weight(p) > 2^{2^d}$, contradiction.

The case when $\mdegs(f) > T$ follows exactly in the same way by using Lemma~\ref{lm:dist-exist-mdegs} instead of Lemma~\ref{lm:dist-exist-mdega}.
\end{proofof}

\newcommand{\fptp}{f_{\PTP}}

\subsection{Proof for Lemma~\ref{lm:dist-exist-mdega}}
Finally we prove Lemma~\ref{lm:dist-exist-mdega}. The proof for Lemma~\ref{lm:dist-exist-mdegs} is completely symmetric using Corollary~\ref{cor:single-dist-exist-mdegs} instead of Lemma~\ref{lm:single-dist-exist-mdega}.
\begin{proofof}{Lemma~\ref{lm:dist-exist-mdega}}
	Recall that a polynomial $p$ is valid, if it is of degree at most $T$, and satisfies $\weight(p) \le 2^{T}$. Let $f_d := \Ada_{f,d}$ and $D_d$ be the domain of $f_d$. We are going to construct these distributions $\distr_0^{d}$'s and $\distr_1^{d}$'s by an elegant induction.
	
	{\bf Construction of $\distr_0$ and $\distr_1$ from Lemma~\ref{lm:single-dist-exist-mdega}.} By Lemma~\ref{lm:single-dist-exist-mdega} there exist two distributions $\distr_0$ and $\distr_1$ supported on $f^{-1}(0)$ and $f^{-1}(1)$ respectively, such that $ - p(\distr_0) > 2 \cdot p(\distr_1) $ for all valid polynomial $p$ computing $f$.
	
	{\bf The base case: construction of $\distr_0^0$ and $\distr_1^0$.}
	For the base case $d=0$, as $f_0$ is just $f$, we simply set $\distr_0^0 = \distr_0$ and $\distr_1^0 = \distr_1$. Then for all valid polynomial $p$ computing $f_0$, we have $- p(\distr_0^0) > 2 \cdot p(\distr_1^0) = 2^{2^0} \cdot p(\distr_1^0)$.
	
	{\bf Construction of $\distr_0^{d}$ and $\distr_1^{d}$ for $d > 0$.}	
	When $d>0$, suppose that we have already constructed the required distributions $\distr_0^{d-1}$ and $\distr_1^{d-1}$ for $f_{d-1}$. Decompose the input to $f_d$ as $(w,x,y) \in D \times D_{d-1} \times D_{d-1}$ as in the definition, we claim that 
	$$\distr_{0}^d = (\ddd{0}{0}{0})\footnote{($\ddd{0}{0}{0})$ is interpreted as the product distribution $\distr_0 \times \distr_{0}^{d-1} \times \distr_0^{d-1}$ on $ D \times D_{d-1} \times D_{d-1}$.} \text{  and  }  \distr_{1}^d = (\ddd{1}{1}{1})$$ satisfy our conditions.
	
	{\bf Analysis of $\distr_0^{d}$ and $\distr_1^{d}$.}
	Note that $D_{i}^d$ is supported on $f_d^{-1}(i)$ for $i \in \{0,1\}$ from the definition. Let $p(w,x,y)$ be a valid polynomial computing $f_d$. We set 
	$$p(\distr_w,\distr_x,\distr_y) := \Ex_{w\sim \distr_w, x \sim \distr_x, y\sim \distr_y}[p(w,x,y)]$$ for simplicity, where $\distr_w,\distr_x,\distr_y$ are distributions over $D,D_{d-1},D_{d-1}$ respectively.
	
	Then we have to verify that for all valid polynomial $p$ computing $f_d$,
	$$
	-p(\distr_0^{d}) = -p(\ddd{0}{0}{0}) > 2^{2^d} \cdot p(\ddd{1}{1}{1}) = 2^{2^d} \cdot p(\distr_1^d).
	$$
	
	We proceed by incrementally changing $(\ddd{0}{0}{0})$ into $(\ddd{1}{1}{1})$, and establish inequalities along the way.
	
	{\bf Step 1: $(\ddd{0}{0}{0}) \Rightarrow (\ddd{0}{1}{0})$.} By the definition, we can see that for any fixed $W \in \support(\distr_0)$ and $Y \in \support(\distr_0^{d-1})$, the polynomial in $x$ defined by $p_{L}(x) := p(W,x,Y)$ is a valid polynomial computing $f_{d-1}$, hence $-p_L(\distr_0^{d-1}) > 2^{2^{d-1}} \cdot p_L(\distr_1^{d-1})$. By linearity, we have
	$$
	-p(\ddd{0}{0}{0}) > 2^{2^{d-1}} \cdot p(\ddd{0}{1}{0}).
	$$
	
	{\bf Step 2: $(\ddd{0}{1}{0}) \Rightarrow (\ddd{1}{1}{0})$.} Similarly, for any fixed $X \in \support(\distr_1^{d-1})$ and $Y \in \support(\distr_0^{d-1})$, by the definition, we can see that the polynomial in $w$ defined by $p_M(w) := -p(w,X,Y)$ is a valid polynomial computing $f$, hence $-p_M(\distr_0) > 2 \cdot p_M(\distr_1)$. Again by linearity, we have
	$$
	p(\ddd{0}{1}{0}) > -2 \cdot p(\ddd{1}{1}{0}) > -p(\ddd{1}{1}{0}).
	$$

	{\bf Step 3: $(\ddd{1}{1}{0}) \Rightarrow (\ddd{1}{1}{1})$.} Finally, for any fixed $W \in \support(\distr_1)$ and $X \in \support(\distr_1^{d-1})$, the polynomial in $y$ defined by $p_{R}(y) := p(W,X,y)$ is a polynomial computing $f_{d-1}$, hence $-p_{R}(\distr_0^{d-1}) > 2^{2^{d-1}} \cdot p_{R}(\distr_1^{d-1})$. By linearity, we have 
	$$
	-p(\ddd{1}{1}{0}) > 2^{2^{d-1}} \cdot p(\ddd{1}{1}{1}).
	$$
	
	Putting the above three inequalities together, we have
	$$
	-p(\distr_0^{d})=-p(\ddd{0}{0}{0}) > 2^{2^d} \cdot p(\ddd{1}{1}{1}) = 2^{2^d} \cdot p(\distr_1^d).
	$$
	
	This completes the proof.
\end{proofof}

\subsection{Application in hardness amplification for polynomial approximation}\label{sec:amply}

In this subsection, we slightly adapt the above proof in order to show Theorem~\ref{theo:ampl}.

For a polynomial $p$ on $n$ variables, let $\|p\|_{\infty} := \max_{x \in \{0,1\}^n} |p(x)|$. Lemma~\ref{lm:dist-exist-mdega} shows that, fix a partial function $f$ with $\mdega(f) > T$, then for any polynomial computing $\Ada_{f,d}$ with $\weight(p) \le 2^{T}$, we must have $\|p\|_{\infty} > 2^{2^d}$. The restriction on $\weight(p)$ is essential for us to establish the connection between $\AzPPdt$ and $\mdega$, but it becomes troublesome when it comes to proving a hardness amplification result.

Luckily, we can get rid of the restriction on $\weight(p)$ by making a stronger assumption that $\dega(f) > T$. Formally, we have the following analogous lemma for Lemma~\ref{lm:dist-exist-mdega}.

\begin{lem}\label{lm:www}
Let $f$ be a partial function with $\dega^{2/3}(f) > T$. Then for each integer $d$, there exist two distributions $\distr_1^d$ and $\distr_0^d$ supported on $\Ada_{f,d}^{-1}(1)$ and $\Ada_{f,d}^{-1}(0)$ respectively, such that for any degree-$T$ polynomial $p$ computing $\Ada_{f,d}$, $- p(\distr_0^d) > 2^{2^d} \cdot p(\distr_1^d)$ and consequently $\|p\|_{+\infty} > 2^{2^d}$.
\end{lem}
\begin{proof}
Using nearly the same proof for Lemma~\ref{lm:single-dist-exist-mdega}, we can show that for a partial function $f$, if $\dega^{2/3}(f) > T$, there exist two distributions $\distr_0$ and $\distr_1$ supported on $f^{-1}(0)$ and $f^{-1}(1)$ respectively, such that
$
- p(\distr_0) > 2 \cdot p(\distr_1)
$
for all degree-$T$ polynomial $p$ computing $f$. Then we can proceed exactly as in the proof for Lemma~\ref{lm:dist-exist-mdega} to get the desired distributions.
\end{proof}

Finally, we are ready to prove Theorem~\ref{theo:ampl}.

\begin{proofof}{Theorem~\ref{theo:ampl}}
	Let $F := \Ada_{f,d}$. Suppose otherwise  $\adeg_{\epsilon}(F) \le T$ for $\epsilon = 0.5-2^{-2^d+1}$. Then there exists a polynomial $p$ such that $\|p\|_{\infty} \le 1 +\epsilon$, $p(x) \le 0.5-2^{-2^d+1}$ when $F(x) = 0$, and $p(x) \ge 0.5+2^{-2^d+1}$ when $F(x) = 1$.
	
	Then we define polynomial $q(x) := (p(x)-0.5) \cdot 2^{2^d-1} $. It is easy to see $q(x)$ computes $F$. Also, we have $\|q\|_{\infty} \le (\|p\|_{\infty} + 0.5) \cdot 2^{2^d-1} < 2^{2^d}$, which contradicts Lemma~\ref{lm:www}, and this completes the proof.
\end{proofof}

	\section{Acknowledgment}
	
	I would like to thank Scott Aaronson, Adam Bouland, Dhiraj Holden and Prashant Vasudevan for several helpful discussions during this work, Ruosong Wang for many comments on an early draft of this paper, Justin Thaler for the suggestion on the application in hardness amplification for polynomial approximation, and Mika G{\"{o}}{\"{o}}s and Thomas Watson for pointing out an issue in the proof of Theorem~\ref{theo:hard-adapt-postbpp}.
	
	\bibliographystyle{plainurl}
	\bibliography{team} 
	
	\newpage
	\appendix
	\section{Preliminaries for the appendix}

\subsection{Conical juntas}
We first introduce the definition for conical juntas (cf.~\cite{goos2015rectangles}), which will be used frequently in this appendix. 

Let $x=x_{1}\ldots x_{M}\in\left\{  0,1\right\}  ^{M}$ be a string. \ Then a
\textit{literal} is a term of the form $x_i$ or $1-x_i$, and a
$k$-\textit{term }is a product of $k$ literals (each involving a different
$x_{i}$), which is $1$ if the literals all take on prescribed values and $0$ otherwise. 

\begin{defi}
	A $T$-conical junta $h$ is a non-negative linear combination of $T$-terms, i.e., $h(x) := \sum_{i} \alpha_i \cdot C_i(x)$, where for each $i$ we have $\alpha_i \ge 0$ and $C_i$ is a $T$-term. We also define $\weight(h) := \sum_i \alpha_i$.
\end{defi}

The following lemma shows that conical juntas are more powerful than randomized decision trees.

\begin{lem}[Essentially Theorem 15 in~\cite{buhrman2002complexity}]\label{lm:conical-junta}
	The acceptance probability of a $T$-query randomized decision tree $\dtree$ can be represented by a $T$-conical junta $h$ with $\weight(h) \le 2^T$.
\end{lem}

\subsection{Complexity classes and their query complexity analogues}

We introduce the complexity classes: \SBP, \SBQP, $\PostBPP$ ($\BPPPATH$) here, and define their analogues in query complexity along the way.

\subsubsection{$\SBP$ and $\SBPdt$}

Now we recall the definition of $\SBP$, there are several equivalent definitions for $\SBP$ in \cite{bohler2006error} (see Proposition 21), we use the most convenient one here.

\begin{defi}
	$\SBP$ (defined by B{\"o}hler, Gla{\ss}er and Meister~\cite{bohler2006error}) is the class of languages $L\subseteq\left\{
	0,1\right\}  ^{\ast}$\ for which there exists a $\mathsf{BPP}$\ machine $M$ and a polynomial $p$, such that for all inputs $x$:
	
	\begin{enumerate}
		\item[(i)] $x\in L\Longrightarrow\Pr\left[  M\left(  x\right)  \text{ accepts}\right]  \geq 2^{-p(|x|)}$.
		
		\item[(ii)] $x\notin L\Longrightarrow\Pr\left[  M\left(  x\right)  \text{
			accepts}\right]  < 2^{-p(|x|)-1}$.
	\end{enumerate}
\end{defi}

Then we define the query complexity analogue of $\SBP$ in the standard way.

\begin{defi}\label{defi:SBPdt}
	Let $f:D \to \{0,1\}$ with $D \subseteq \{0,1\}^M$ be a partial function. We say a randomized decision tree $\dtree$ $\SBP$-computes $f$ if
	$$
	\Pr[\dtree(x) = 1] > 2 \cdot \Pr[\dtree(y) = 1]
	$$
	
	for all $x \in f^{-1}(1)$ and $y \in f^{-1}(0)$. 
	
	We define $\SBPdt(f)$ as the minimum of $\mathcal{C}(\dtree)$ over all $\dtree$ $\SBP$-computing $f$.
	
	And we simply let $\coSBPdt(f) := \SBPdt(\neg f)$.
\end{defi}

It may seem strange at first that there is no $\log_2(1/\alpha)$ term in our definition of $\SBPdt(f)$. Actually, one can show that having the $\log_2(1/\alpha)$ term or not would not change the partial function class $\SBPdt$: the following lemma shows that whenever we have a randomized decision tree $\mathcal{T}$ $\SBP$-computing a function $f$, $\mathcal{T}$ can be made to $\SBP$-compute $f$ with a reasonable probability gap.

\begin{lem}[Proposition 33 in~\cite{goos2015rectangles}]
	Let $f:D \to \{0,1\}$ with $D \subseteq \{0,1\}^M$ be a partial function. Suppose $d=\SBPdt(f)$. Then there is a randomized decision tree $\dtree$ $\SBP$-computing $f$ and a real number $\alpha$, such that
	$$
	\Pr[\dtree(x) = 1] > 2 \cdot \alpha \text{ and } \Pr[\dtree(y) = 1] \le \alpha \text{ and } \alpha \ge 2^{-(d+1)} \binom{n}{d}^{-1}
	$$
	for all $x \in f^{-1}(1)$, $y \in f^{-1}(0)$.
\end{lem}
\subsubsection{$\PostBPP$ and $\PostBPPdt$}

In this subsection we review the definition of $\PostBPP$, and define its analogue in query complexity.

Roughly speaking, $\PostBPP$ consists of the computational problems can be solved in probabilistically polynomial time, given the ability to {\em postselect} on an event (which may happen with a very small probability). Formally:

\begin{defi}
	$\PostBPP$ (defined by Han, Hemaspaandra, and Thierauf \cite{han1997threshold}\footnote{In the original paper it is called $\BPPPATH$.}) is the class of languages $L\subseteq\left\{
	0,1\right\}  ^{\ast}$\ for which there exists a $\mathsf{BPP}$\ machine $M$,
	which can either \textquotedblleft succeed\textquotedblright\ or
	\textquotedblleft fail\textquotedblright\ and conditioned on succeeding either
	\textquotedblleft accept\textquotedblright\ or \textquotedblleft
	reject,\textquotedblright\ such that for all inputs $x$:
	
	\begin{enumerate}
		\item[(i)] $\Pr\left[  M\left(  x\right)  \text{ succeeds}\right]  >0$.
		
		\item[(ii)] $x\in L\Longrightarrow\Pr\left[  M\left(  x\right)  \text{ accepts
		}|~M\left(  x\right)  \text{ succeeds}\right]  \geq\frac{2}{3}$.
		
		\item[(iii)] $x\notin L\Longrightarrow\Pr\left[  M\left(  x\right)  \text{
			accepts }|~M\left(  x\right)  \text{ succeeds}\right]  \leq\frac{1}{3}$.
	\end{enumerate}
\end{defi}

$\PostBPPdt(f)$ can be defined similarly.

\begin{defi}\label{defi:PostBPPdt}
	Now we allow a randomized decision tree to output a failure mark $*$ besides $0$ and $1$.
	
	Let $f:D \to \{0,1\}$ with $D \subseteq \{0,1\}^M$ be a partial function. We say a randomized decision tree $\dtree$ $\PostBPP$-computes $f$ if
	
	$$
	\Pr[\dtree(x) = 1] \ge 2 \cdot \Pr[\dtree(x) = 0] \text{ and } \Pr[\dtree(y) = 0] \ge 2 \cdot \Pr[\dtree(y) = 1]
	$$
	
	for all $x \in f^{-1}(1)$ and $y \in f^{-1}(0)$. 
	
	Fix a $\dtree$ $\PostBPP$-computing $f$, let $\alpha$ be the maximum real number such that
	$$
	\Pr[\dtree(x) \ne *] \ge \alpha
	$$
	
	for all $x \in D$.
	
	Then we define $\PostBPPdt(\dtree;f) = C(\dtree) + \log_2(1/\alpha)$ for $\dtree$ $\PostBPP$-computing $f$, and $\PostBPPdt(f)$ as the minimum of $\PostBPPdt(\dtree;f)$ over all $\dtree$ $\PostBPP$-computing $f$.
\end{defi}

\subsubsection{$\SBQP$ and $\SBQPdt$}
In this subsection we review the definition of $\SBQP$, and define its analogue in query complexity. Roughly speaking, $\SBQP$ is just the quantum analogue of $\SBP$.

\begin{defi}
	$\SBQP$ (defined by Kuperberg~\cite{kuperberg2009hard}) is the class of languages $L\subseteq\left\{
	0,1\right\}  ^{\ast}$\ for which there exists a polynomial-time quantum algorithm $M$ and a polynomial $p$, such that for all inputs $x$:
	
	\begin{enumerate}
		\item[(i)] $x\in L\Longrightarrow\Pr\left[  M\left(  x\right)  \text{ accepts}\right]  \geq 2^{-p(|x|)}$.
		
		\item[(ii)] $x\notin L\Longrightarrow\Pr\left[  M\left(  x\right)  \text{
			accepts}\right]  \leq 2^{-p(|x|)-1}$.
	\end{enumerate}
\end{defi}

Then we define its query complexity analogue.

\begin{defi}\label{defi:SBQPdt}
	Let $f:D \to \{0,1\}$ with $D \subseteq \{0,1\}^M$ be a partial function. We say a quantum query algorithm $\qalg$ $\SBQP$-computes $f$ if
	
	$$
	\Pr[\qalg(x) = 1] \ge 2 \cdot \Pr[\qalg(y) = 1] \text{ and } 
	\Pr[\qalg(x) = 1] > 0
	$$
	
	for all $x \in f^{-1}(1)$ and $y \in f^{-1}(0)$. 
	
	Fix a $\qalg$ $\SBQP$-computing $f$, let $\alpha$ be the maximum real number such that
	
	$$
	\Pr[\qalg(x) = 1] \ge 2\alpha \text{ and } \Pr[\qalg(y) = 1] \le \alpha
	$$
	
	for all $x \in f^{-1}(1)$ and $y \in f^{-1}(0)$. 
	
	Then we define $\SBQPdt(\qalg;f) = C(\qalg) + \log_2(1/\alpha)$ for $\qalg$ $\SBQP$-computing $f$ and $\SBQPdt(f)$ as the minimum of $\SBQPdt(\qalg;f)$ over all $\qalg$ $\SBQP$-computing $f$.
	
	And we simply let $\coSBQPdt(f) := \SBQPdt(\neg f)$.
\end{defi}

\subsubsection{$\PostBQP$ and $\PostBQPdt$}
$\PostBQP$ is defined similarly as $\PostBPP$, just replaced the $\BPP$ machine by a polynomial time quantum algorithm. And $\PostBQPdt(f)$ is defined in the same way as $\PostBPP(f)$ except for changing the randomized decision tree $\dtree$ to a quantum query algorithm $\qalg$.

\section{Missing proofs in Section~\ref{sec:pre}}

\begin{proofof}{Theorem~\ref{theo:eq-AZPP-mdega}}
	For the first claim, suppose $\AzPPdt(f)=d$, then there exists a $T$-query randomized decision tree $\dtree$ and a constant $\alpha>0$, such that
	\begin{itemize}
		\item $\Pr[\dtree(x)=1] \ge 1/2$ for all $x \in D$.
		\item $\Pr[\dtree(x)=1] - 1/2 \ge 2 \alpha$ and $\Pr[\dtree(y)=1] - 1/2 \le \alpha$ for all $x \in f^{-1}(1)$ and $y \in f^{-1}(0)$.
		\item $T$ + $\log_2(1/\alpha) = d$.
	\end{itemize}
	
	Let $h$ be the conical junta representing the acceptance probability of $\dtree$, we have $\weight(h) \le 2^T$ by Lemma~\ref{lm:conical-junta}.
	
	By expanding every $T$-term into $2^T$ monomials, we can further represent $h$ by a polynomial $p_h$ with $\weight(p_h) \le 2^{2T}$.
	
	Now, we define the polynomial 
	$$p(x) := \frac{1}{2\alpha} \cdot (p_h(x)-1/2).$$ 
	
	We claim that $p$ one-sided approximates $f$. Indeed, when $f(x)=0$, we have $p_h(x) \in [1/2,1/2+\alpha]$, hence $p(x) \in [0,1/2]$; and when $f(x) = 1$, we have $p_h(x) \ge 1/2+2\alpha$, hence $p(x) \ge 1$.
	
	Moreover, 
	$$\weight(p) \le (\weight(p_h) + 1/2) \cdot \frac{1}{2\alpha} \le 2^{2T}/\alpha,$$
	
	the last inequality holds as $\alpha < 1/4$.
	
	Hence 
	\begin{align*}
	\qquad\qquad\qquad\qquad\mdega(f) &\le \max \{\deg(p),\log_2(\weight(p))\} \\&\le \max\{T, 2T + \log_2(1/\alpha)\} \le 2d = 2 \cdot \AzPPdt(f).
	\end{align*}
	
	For the second claim, suppose $\mdega(f)=d$, then there exists a $T$-degree polynomial $p$ one-sided approximating $f$ such that $T \le d$ and $\weight(p) \le 2^d$. 
	
	Let $p(x) = \sum_{i=1}^{m} a_i \cdot M_i(x)$ and $S = \weight(p) = \sum_{i=1} |a_i|$, such that for each $i$, $a_i \in \R$ and $M_i$ is a unit monomial (i.e., $M_i(x) := x_{i_1}x_{i_2}\cdots x_{i_k}$).
	
	Consider the following algorithm:
	
	\begin{itemize}
		\item Pick an integer $i \in [m]$ by selecting $j$ with probability $|a_j| / S$.
		\item Query all the variables involved in $M_i$ to calculate $M_i(x)$.
		\item If $M_i(x) = 1$, accept if $a_i > 0$ and reject otherwise.
		\item If $M_i(x) = 0$, accept with probability $1/2$.
	\end{itemize}
	
	Clearly, as $p$ is of degree $T$, the above algorithm can be implemented by a $T$-query randomized decision tree $\dtree$. 
	
	Now we analyze the acceptance probability of $\dtree$ on an input $x$. We can see 
	$$
	\Pr[\dtree(x) = 1] = \sum_{i} \frac{1+M_i(x) \cdot a_i/|a_i|}{2} \cdot \frac{|a_i|}{S} = \frac{1}{2} + p(x) \cdot \frac{1}{2S}.
	$$
	
	Which means, when $f(x) = 0$, we have $p(x) \in [0,1/2]$, hence $\frac{1}{2} \le \Pr[\dtree(x) = 1] \le \frac{1}{2} + \frac{1}{4S}$; and when $f(x) = 1$, we have $p(x) \ge 1$, therefore $\frac{1}{2} \le \Pr[\dtree(x) = 1] \le \frac{1}{2} + \frac{1}{2S}$. So we can take $\alpha = \frac{1}{4S}$ and we have
	$$\AzPPdt(f) \le \AzPPdt(f;\dtree) \le T +　\log_2(4S) = T + 2 + \log_2(\weight(p)) \le 2 \cdot \mdega(f) + 2.$$
	
	This completes the proof.
	
\end{proofof}

\section{Missing proofs in Section~\ref{sec:quantum}}

\begin{proofof}{Lemma~\ref{lm:help-1}}
	Suppose not, let $p$ be a degree-$T$ polynomial computing $f$, and satisfies $\weight(p) \le 2^{T}$ and $\max_{ x \in f^{-1}(0)} -p(x) \le 2 \cdot \min_{y \in f^{-1}(1)} p(y)$.
	
	Let $C = \max_{ x \in f^{-1}(0)} -p(x)$, consider the following polynomial 
	$$q(x) := \frac{2}{3}\cdot (p(x)/C + 1).$$ 
	We can see that when $f(x) = 0$, we have $p(x) \in [-C,-1]$, hence $q(x) \in [0, \frac{2}{3}]$; and when $f(x) = 1$, we have $p(x) \ge \frac{1}{2} \cdot C$, therefore $q(x) \ge 1$. Which means $q$ one-sided approximates $f$ with error constant $2/3$.
	
	Also, we have $\weight(q) \le \weight(p)\cdot \frac{2}{3C} + \frac{2}{3} \le \weight(p)$ as $C \ge 1$. So $\max\{\mathrm{deg}(q),\log_2(\weight(q))\} \le T$, contradiction to the fact that $\mdega^{2/3}(f) > T$.
\end{proofof}

\begin{proofof}{Lemma~\ref{lm:pps-cons}}
	By our assumption, there exists a $t$-query randomized decision tree $\dtree$ and a real number $\alpha > 0$ such that
	
	\begin{itemize}
		\item when $f(x) = 1$, $\Pr[\dtree(x) = 1] \ge \frac{1}{2} + \alpha$.
		\item when $f(x) = 0$, $\Pr[\dtree(x) = 1] \le \frac{1}{2} - \alpha$.
		\item $t + \log_2(1/\alpha) \le T$.
	\end{itemize}
	
	Let $h$ be the conical junta representing the accepting probability of $\dtree$. We have $\weight(h) \le 2^t$.
	
	By expanding every $t$-term into $2^t$ monomials, we can further represent $h$ by a polynomial $p_h$ with $\weight(p_h) \le 2^{2t}$. Now we define $p(x) := (p_h(x)-\frac{1}{2})/\alpha$. Clearly $p$ computes $f$.
	
	Moreover, $\weight(p) \le (2^{2t} + \frac{1}{2} ) \cdot (1/\alpha) \le 2^{2T}$, which completes the proof.
	
\end{proofof}
	\section{Proof for the classical case}\label{sec:classical}
In this section we prove Theorem~\ref{theo:ada-postbpp} and Theorem~\ref{theo:hard-adapt-postbpp}.

\subsection{$\SBPdt$ by conical juntas}

We first show when considering the $\SBPdt$, we can work with a {\em conical junta} instead of a randomized decision tree.

\begin{prop}\label{prop:equi-junta}
The definition of $\SBPdt(f)$ is unchanged if we replace the $T$-query randomized decision tree by a $T$-conical junta.
\end{prop}
\begin{proof}
We are going to show the existence of a $T$-query randomized decision tree $\dtree$ $\SBP$-computing $f$ is equivalent to the existence of a $T$-conical junta $h$ $\SBP$-computing $f$.
	
Suppose there exists a $T$-query randomized decision tree $\dtree$ $\SBP$-computing $f$, then the acceptance probability of $\dtree$ can be presented as a $T$-conical junta by Lemma~\ref{lm:conical-junta}.

For the other direction, suppose there exists a $T$-conical junta $h$ $\SBP$-computing $f$, let $h(x) := \sum_{i} \alpha_i \cdot C_i(x)$. Consider the following algorithm: let $P = \sum_i \alpha_i$, we pick a random $T$-term by selecting $C_i$ with probability $\alpha_i/P$ and accept if $C_i$ evaluates to $1$ on the given input. It is not hard to see the above algorithm can be represented by a $T$-query randomized decision tree, and it $\SBP$-computes $f$.
\end{proof}

\subsection{A dual characterization for $\SBPdt$}

We first establish an equivalent dual condition of a function having large $\SBPdt$ complexity.

\begin{lem}\label{lm:equivalent-cond}
	Let $f:D \to \{0,1\}$ with $D \subseteq \{0,1\}^M$ be a partial function, $T$ be a positive integer, $\SBPdt(f) > T$ if and only if there exist two distributions $\distr_0$ and $\distr_1$ supported on $f^{-1}(0)$ and $f^{-1}(1)$ respectively, such that
	
	$$
	C(\distr_0) \ge \frac{1}{2} \cdot C(\distr_1) \text{ for any $T$-term $C$,}
	$$
	
	where $C(\distr_i)$ is defined as $\Ex_{x \sim \distr_i}[C(x)]$ for $i \in \{0,1\}$.
\end{lem}
\begin{proof}
	let $\mathcal{H}_T$ be the set of all $T$-conical juntas on $\{0,1\}^M$, and $f^i:=f^{-1}(i)$ for $i \in \{0,1\}$, by Proposition~\ref{prop:equi-junta}, $\SBPdt(f) > T$ is equivalent to 
	
	$$
	\min_{h \in \mathcal{H}_T} \max_{(x,y) \in f^0\times f^1} \left(h(x) - \frac{1}{2} \cdot h(y) \right) \ge 0.
	$$
	
	Then by the minimax theorem, the above is again equivalent to
	
	$$
	\max_{\distr_{xy} \text{ on }f^0\times f^1} \min_{h \in \mathcal{H}_T}  \Ex_{(x,y) \sim \distr_{xy}}\left(h(x) - \frac{1}{2} \cdot h(y) \right) \ge 0.
	$$
	
	where $\distr_{xy}$ is a distribution on $f^0 \times f^1$.
	
	Observe that we can further take $\distr_{xy}$ to be a product distribution and we can assume $h$ is just a $T$-term. Putting everything together, $\SBPdt(f) > T$ is equivalent to
	
	$$
	\max_{\distr_0 \text{ on } f^0} \max_{\distr_1 \text{ on } f^1} \min_{C\text{ is a $T$-term}} \Ex_{x \sim \distr_0, y \sim \distr_1} \left(C(x) - \frac{1}{2} \cdot C(y) \right) \ge 0,
	$$
	where $\distr_i$ is a distribution on $f^i$ for $i \in \{0,1\}$. This completes the proof.
	
\end{proof}

\begin{rem}
	Another way to prove the above lemma is to use strong duality in linear programming directly. We feel that our proof by minimax theorem is conceptually cleaner.
\end{rem}

The following corollary follows from the definition.

\begin{cor}\label{cor:equivalent-cond-coSBP}
	Let $f:D \to \{0,1\}$ with $D \subseteq \{0,1\}^M$ be a partial function, $T$ be a positive integer, $\coSBPdt(f) > T$ if and only if there exist two distributions $\distr_0$ and $\distr_1$ supported on $f^{-1}(0)$ and $f^{-1}(1)$ respectively, such that
	
	$$
	C(\distr_1) \ge \frac{1}{2} \cdot C(\distr_0) \text{ for any $T$-term $C$.}
	$$
\end{cor}

\subsection{Proof for Theorem~\ref{theo:ada-postbpp} and Theorem~\ref{theo:hard-adapt-postbpp}}

We first show Theorem~\ref{theo:hard-adapt-postbpp} implies Theorem~\ref{theo:ada-postbpp}.

\begin{proofof}{Theorem~\ref{theo:ada-postbpp}}
Suppose $\funcfami{f} \not\in \SBPdt$, the case that $\funcfami{f} \not\in \coSBPdt$ is similar.

Then there exists a sequence of function $\{f_i\}_{i=1}^{\infty} \subseteq \funcfami{f}$ such that $\SBPdt(f_i) > \log(\inputlen(f_i))^{i}$. Then we consider the partial function sequence $\{ \Ada_{f_i,\lceil\log(\inputlen(f_i)) \rceil} \}_{i=1}^{\infty} \subseteq \Ada_{\funcfami{f}}$.

By Theorem~\ref{theo:hard-adapt-postbpp}, we have 
$$
\PostBPPdt(\Ada_{f_i,\lceil\log(\inputlen(f_i)) \rceil}) > \min(\log(\inputlen(f_i))^{i}/5,(\inputlen(f_i)-1)/5).
$$

Note that $\inputlen(\Ada_{f_i,\lceil\log(\inputlen(f_i)) \rceil}) \le 2 \cdot \inputlen(f_i)^2$, we can see $\Ada_{\funcfami{f}} \notin \PostBPPdt$ due to the above sequence.
\end{proofof}
	
Now we are going to prove Theorem~\ref{theo:hard-adapt-postbpp}. We say a pair of conical juntas $a(x)$ and $r(x)$ computes a function $f$ if it satisfies the following two conditions.
	
\begin{itemize}
	\item When $f(x)=1$, $a(x) \ge 5\cdot r(x)$ and $a(x) \ge 1$.
	\item When $f(x)=0$, $r(x) \ge 5\cdot a(x)$ and $r(x) \ge 1$.
\end{itemize}

In order to lower bound the $\PostBPPdt$ complexity of some functions, we introduce some consequences of a function having low $\PostBPPdt$ complexity.
\begin{lem}\label{lm:postbpp-cons}
	Let $f:D \to \{0,1\}$ with $D \subseteq \{0,1\}^M$ be a partial function, $T$ be a positive integer. Suppose $\PostBPPdt(f) \le T$, then there exist two $5T$-conical juntas $a(x)$ and $r(x)$ such that
	
	\begin{itemize}
		\item The pair of $a(x)$ and $r(x)$ computes $f$.
		\item $\max_{x \in \{0,1\}^M} a(x) \le 2^{5T+1}$ and $\max_{x \in \{0,1\}^M} r(x) \le 2^{5T+1}$.
	\end{itemize}
\end{lem}

\begin{proof}
	Amplifying the probability gap by taking the majority of $5$ independent runs, we get a randomized decision tree $\dtree$ such that
	\begin{itemize}
		\item $\PostBPPdt(\dtree;f) \le 5T$.
		\item 	$
		\Pr[\dtree(x) = 1] \ge 5 \cdot \Pr[\dtree(x) = 0] \text{ and } \Pr[\dtree(y) = 0] \ge 5 \cdot \Pr[\dtree(y) = 1]
		$ for all $x \in f^{-1}(1)$ and $y \in f^{-1}(0)$. 
	\end{itemize}
	
	Then we simply define $a(x)$ ($r(x)$) as $2^{5T+1}$ multiplies the acceptance (reject) probability of $\dtree$. $a(x)$ and $r(x)$ can be represented by $5T$-conical juntas by Lemma~\ref{lm:conical-junta}.
	
	Now we show $a(x)$ and $r(x)$ satisfy our conditions. The second condition follows directly from their definitions. 
	For the first condition, when $f(x)=1$, we have $a(x) \ge 5 \cdot r(x)$ by their definitions, and since $\Pr[\dtree(x) \in \{0,1\}] \ge 2^{-5T}$ for all $x \in D$, $a(x) \ge 2^{5T+1} \cdot \frac{5}{6} \cdot 2^{-5T} \ge 1$. The case when $f(x)=0$ can be verified in the same way, and this completes the proof.
\end{proof}
	
Our proof proceed by a similar fashion as in Section~\ref{sec:quantum}, it again relies on the following two key lemmas.

\begin{lem}\label{lm:dist-exist-SBP}
Let $f:D \to \{0,1\}$ with $D \subseteq \{0,1\}^M$ be a partial function with $\SBPdt(f) > T$. Then for each integer $d$, there exist two distributions $\distr_1^d$ and $\distr_0^d$ supported on $\Ada_{f,d}^{-1}(1)$ and $\Ada_{f,d}^{-1}(0)$ respectively, such that $r(\distr_0^d)> 2^{2^d} \cdot a(\distr_1^d)$ for any $T$-conical juntas $a(x)$ and $r(x)$ computing $\Ada_{f,d}$.
\end{lem}

\begin{lem}\label{lm:dist-exist-coSBP}
Let $f:D \to \{0,1\}$ with $D \subseteq \{0,1\}^M$ be a partial function with $\coSBPdt(f) > T$. Then for each integer $d$, there exist two distributions $\distr_1^d$ and $\distr_0^d$ supported on $\Ada_{f,d}^{-1}(1)$ and $\Ada_{f,d}^{-1}(0)$ respectively, such that $a(\distr_1^d) > 2^{2^d} \cdot a(\distr_1^d)$ for any $T$-conical juntas $a(x)$ and $r(x)$ computing $\Ada_{f,d}$.
\end{lem}

Before proving Lemma~\ref{lm:dist-exist-SBP} and Lemma~\ref{lm:dist-exist-coSBP}, we show they imply Theorem~\ref{theo:hard-adapt-postbpp}.

\begin{proofof}{Theorem~\ref{theo:hard-adapt-postbpp}}
We first prove the case $\SBPdt(f) > T$. Suppose  $\PostBPPdt(\Ada_{f,d}) \le \min(T/5,(2^{d}-1)/5)$, by Lemma~\ref{lm:postbpp-cons}, there is a pair of $T$-conical juntas $a(x)$ and $r(x)$ computing $\Ada_{f,d}$ such that $\max_{x} r(x) \le 2^{2^d}$.

By Lemma~\ref{lm:dist-exist-SBP}, we have $r(\distr_0^d) > 2^{2^d} \cdot a(\distr_1^d) \ge 2^{2^d}$. Hence there must exist an $x$ such that $r(x) > 2^{2^d}$, contradiction.

Then case for $\coSBPdt(f) > T$ follows from exactly the same argument and Lemma~\ref{lm:dist-exist-coSBP}.
\end{proofof}

\subsection{Proof for Lemma~\ref{lm:dist-exist-SBP}}

Finally we prove Lemma~\ref{lm:dist-exist-SBP}, the proof for Lemma~\ref{lm:dist-exist-coSBP} is completely symmetric.

\begin{proofof}{Lemma~\ref{lm:dist-exist-SBP}}
	We are going to construct those distributions by an induction on $d$. Let $f_d := \Ada_{f,d}$ and $D_d$ be the domain of $f_d$.
	
	{\bf Construction of $\distr_0$ and $\distr_1$ from Lemma~\ref{lm:equivalent-cond}.}
	As $\SBPdt(f) > T$, by Lemma~\ref{lm:equivalent-cond}, there exist two distributions $\distr_0$ and $\distr_1$ supported on $f^{-1}(0)$ and $f^{-1}(1)$ respectively, such that
	$$
	h(\distr_0) \ge \frac{1}{2} \cdot h(\distr_1) \text{ for any $T$-conical junta $h$.}
	$$
	
	{\bf The base case: construction of $\distr_0^0$ and $\distr_1^0$.}	
	The base case $d=0$ is very simple. $f_0$ is just the $f$ itself. We let $\distr_1^0 = \distr_1$ and $\distr_0^0=\distr_0$. Then we have $a(\distr_0) \ge \frac{1}{2} \cdot a(\distr_1)$. Also,  $r(\distr_0)\ge 5 \cdot a(\distr_0)$ as $\distr_0$ is supported on $f^{-1}(0)$. Putting these facts together, we have $r(\distr_0) > 2 \cdot a(\distr_1)$, which completes the case for $d = 0$.
	
	{\bf Construction of $\distr_0^0$ and $\distr_1^0$ for $d >　0$.}	
	For $d>0$, suppose that we have already constructed distributions $\distr_0^{d-1}$ and $\distr_1^{d-1}$ on inputs of $f_{d-1}$. Decompose the input to $f_d$ as $(w,x,y) \in D \times D_{d-1} \times D_{d-1}$ as in the definition, we claim that 
	$$\distr_{0}^d = (\ddd{0}{0}{0})\footnote{recall that ($\ddd{0}{0}{0})$ is interpreted as the product distribution $\distr_0 \times \distr_{0}^{d-1} \times \distr_0^{d-1}$ on $ D \times D_{d-1} \times D_{d-1}$.}$$ and $$\distr_{1}^d = (\ddd{1}{1}{1})$$ satisfy our conditions.
	
	{\bf Analysis of $\distr_0^{d}$ and $\distr_1^{d}$.}
	Note that $D_{i}^d$ is supported on $f_d^{-1}(i)$ for $i \in \{0,1\}$ from the definition. For a conical juntas $h(w,x,y)$ on $D \times D_{d-1} \times D_{d-1}$, We let 
	$$h(\distr_w,\distr_x,\distr_y) := \Ex_{w\sim \distr_w, x \sim \distr_x, y\sim \distr_y}[h(w,x,y)]$$ for simplicity, where $\distr_w,\distr_x,\distr_y$ are distributions over $D,D_{d-1},D_{d-1}$ respectively. 
	
	We have to verify that
	$$
	r(\distr_0^d) = r(\distr_0,\distr_0^{d-1},\distr_0^{d-1}) > 2^{2^d} \cdot a(\distr_1,\distr_1^{d-1},\distr_1^{d-1}) = 2^{2^d} \cdot a(\distr_1^d).
	$$
	
	In the same way as in Section~\ref{sec:quantum}, we proceed by incrementally changing $(\ddd{0}{0}{0})$ into $(\ddd{1}{1}{1})$.

	{\bf Step 1: $(\ddd{0}{0}{0}) \Rightarrow (\ddd{0}{1}{0})$.}
	Consider the following two $T$-conical juntas on $x$: 
	$$
	a_L(x) := a(\distr_0,x,\distr_0^{d-1}) =\Ex_{w \sim \distr_0,y \sim \distr_0^{d-1}} [a(w,x,y)],
	$$ and
	$$
	r_L(x) := r(\distr_0,x,\distr_0^{d-1}) = \Ex_{w \sim \distr_0,y \sim \distr_0^{d-1}} [r(w,x,y)].
	$$
	
	Note that for any fixed $W \in \support(\distr_0)$ and any $Y \in \support(\distr_0^{d-1})$, by the definition of $f_d$, the $T$-conical junta pair $a(W,x,Y)$ and $r(W,x,Y)$ must compute $f_{d-1}$. It is not hard to verify by linearity, that their expectations $a_L(x)$ and $r_L(x)$ also compute $f_{d-1}$.
	
	Therefore, plugging in $\distr_0^{d-1}$ and $\distr_1^{d-1}$, we have $r_L(\distr_0^{d-1}) > 2^{2^{d-1}} \cdot a_L(\distr_0^{d-1})$, which means
	$$r(\distr_0,\distr_0^{d-1},\distr_0^{d-1}) > 2^{2^{d-1}} \cdot a(\distr_0,\distr_1^{d-1},\distr_0^{d-1}).$$
	
	{\bf Step 2: $(\ddd{0}{1}{0}) \Rightarrow (\ddd{1}{1}{0})$.} Then, for each fixed $X,Y$, the polynomial $a_M(w) := a(w,X,Y)$ is a $T$-conical junta, so we have $a_M(\distr_0) \ge \frac{1}{2} \cdot a_M(\distr_1)$. Hence by linearity, 
	$$
	a(\distr_0,\distr_1^{d-1},\distr_0^{d-1}) \ge \frac{1}{2} \cdot a(\distr_1,\distr_1^{d-1},\distr_0^{d-1}).
	$$
	
	Now, notice that $\distr_1$ is supported on $f^{-1}(1)$, and $\distr_0^{d-1}$ is supported on $f_{d-1}^{0}$, so $(\distr_1,\distr_1^{d-1},\distr_0^{d-1})$ is supported on $f_d^{-1}(1)$, therefore 
	$$
	a(\distr_1,\distr_1^{d-1},\distr_0^{d-1}) \ge 5 \cdot r(\distr_1,\distr_1^{d-1},\distr_0^{d-1}).
	$$ 
	
	{\bf Step 3: $(\ddd{1}{1}{0}) \Rightarrow (\ddd{1}{1}{1})$.} Finally, consider the polynomials on $y$ defined by
	$$
	a_{R}(y) := a(\distr_1,\distr_1^{d-1},y)
	\text{ and } 
	r_{R}(y) := r(\distr_1,\distr_1^{d-1},y).
	$$
	By the same augment as above, they are also a pair of $T$-conical juntas computing $f_{d-1}$, so plugging in $\distr_0^{d-1}$ and $\distr_1^{d-1}$ again, we have $r_{R}(\distr_0^{d-1}) > 2^{2^{d-1}} \cdot a_{R}(\distr_1^{d-1})$, which means 
	$$
	r(\distr_1,\distr_1^{d-1},\distr_0^{d-1}) > 2^{2^{d-1}} \cdot a(\distr_1,\distr_1^{d-1},\distr_1^{d-1}).
	$$
	
	Putting everything together, we have
	$$
	r(\distr_0,\distr_0^{d-1},\distr_0^{d-1}) > (2^{2^{d-1}} \cdot \frac{1}{2} \cdot 5 \cdot 2^{2^{d-1}}) \cdot a(\distr_1,\distr_1^{d-1},\distr_1^{d-1}) > 2^{2^d} \cdot a(\distr_1,\distr_1^{d-1},\distr_1^{d-1}).
	$$
	
	This completes the proof.
\end{proofof}

	\section{Formal proofs for the oracle separations}\label{sec:oracle}

\paragraph*{$\PTIME^{\NP^{\Oracle}} \not\subset \PP^{\Oracle}$}
We begin with a famous lower bound on $\dega(\AND_n)$ by Nisan and Szegedy.

\begin{theo}[\cite{nisan1994degree}]
$\dega(\AND_n) \ge \Omega(\sqrt{n})$.
\end{theo}

Then we consider the problem $\adaand_n := \Ada_{\AND_n,\log_2 n}$.

By Theorem~\ref{theo:hard-adapt-postbqp}, we have

$$
\PPdt(\adaand_n) \ge \Omega(\sqrt{n}).
$$

On the other hand, there is a simple $\polylog(n)$-time $\PTIME^{\NP}$ algorithm for $\adaand$. By a standard diagonalization argument, we have the following corollary.

\begin{cor}
There exists an oracle $\Oracle$ such that $\PTIME^{\NP^{\Oracle}} \not\subset \PP^{\Oracle}$.
\end{cor}

\paragraph*{$\PTIME^{\SZK^{\Oracle}} \not\subset \PP^{\Oracle}$}
In order to establish the oracle separation between $\PTIME^{\SZK}$ and $\PP$, we need the following results in \cite{aaronson2012impossibility}.

\begin{theo}[Essentially Theorem~8 in~\cite{aaronson2012impossibility}]\label{theo:PTP-lowb}
	$\dega(\PTP_n) \ge \Omega(n^{1/3})$.
\end{theo}

\begin{prop}[Proposition 2 in~\cite{aaronson2012impossibility}]\label{theo:PTP-easy}
	$\PTP_n$ has an $O(\log n)$ time $\SZK$ protocol.
\end{prop}

Then for the problem $\adaptp_n := \Ada_{\PTP_n,\log_2 n}$, by Theorem~\ref{theo:PTP-lowb} and Theorem~\ref{theo:ada-postbqp}, we have
$$
\PPdt(\adaptp_n) \ge \Omega(n^{1/3}).
$$

By Proposition~\ref{theo:PTP-easy}, we can see $\adaptp$ admits a $\polylog(n)$-time $\PTIME^{\SZK}$ algorithm, hence again by a standard diagonalization argument, we have the following corollary.

\begin{cor}
	There exists an oracle $\Oracle$ such that $\PTIME^{\SZK^{\Oracle}} \not\subset \PP^{\Oracle}$.
\end{cor}
	
\end{document}